\documentclass[11pt]{article}
\usepackage{amsthm,amsfonts,amsmath,times}
\usepackage{enumerate}
\usepackage{algo, extras, xspace, color}
\usepackage{boxedminipage, wrapfig}
\usepackage{hyperref,url}

\def\Comment#1{\textsl{$\langle\!\langle$#1\/$\rangle\!\rangle$}}

\advance \topmargin by -1.0in
\advance \oddsidemargin by -0.8in
\newcommand{\myparskip}{3pt}
\parskip \myparskip

\newtheorem{lemma}{Lemma}[section]
\newtheorem{theorem}[lemma]{Theorem}

\newtheorem{corollary}[lemma]{Corollary}
\newtheorem{prop}[lemma]{Proposition}
\newtheorem*{claim}{Claim}
\newtheorem{remark}[lemma]{Remark}

\newcommand{\etal}{{\em et al.}~}
\DeclareMathOperator*{\Ex}{\mathbb{E}}

\renewenvironment{proof}{\vspace{-0.1in}\noindent{\bf Proof:}}%
        {\hspace*{\fill}$\Box$\par}
\newenvironment{proofof}[1]{\smallskip\noindent{\bf Proof of #1:}}%
        {\hspace*{\fill}$\Box$\par}
        {\hspace*{\fill}$\Box$\par}

\def\script#1{\mathcal{#1}}
\def\bx{\textbf{x}}
\def\tbx{\tilde{\textbf{x}}}
\def\sp{\;|\;}
\def\b1{\textbf{1}}

\def\MCSAfull{Minimum Submodular-Cost Allocation\xspace}
\def\MCSA{\textsc{MSCA}\xspace}

\def\SubMP{\textsc{Sub-MP}\xspace}
\def\kSubMP{\textsc{k-way-Sub-MP}\xspace}
\def\kHMC{\textsc{k-way-Hypergrap-Cut}\xspace}
\def\SubMPRel{\textsc{SubMP-Rel}\xspace}

\def\SymSubMP{\textsc{Sym-Sub-MP}\xspace}
\def\kSymSubMP{\textsc{k-way-Sym-Sub-MP}\xspace}

\def\HMC{\textsc{Hypergraph-MC}\xspace}

\def\AHMC{\textsc{Hypergraph-MC}\xspace}

\def\HMP{\textsc{Hypergraph-MP}\xspace}

\def\SymMPR{\textsc{SymSubMP-Rounding}\xspace}
\def\SubMPRH{\textsc{SubMP-Half-Rounding}\xspace}

\def\nodeMC{\textsc{Node-wt-MC}\xspace}
\def\MC{\textsc{Graph-MC}\xspace}
\def\optcr{{\textsc{OPT}_{\textsc{frac}}}\xspace}
\def\lovasz{Lov\'{a}sz\xspace}

\begin{document}
\title{Approximation Algorithms for Submodular Multiway Partition}

\author{
Chandra Chekuri
\thanks{
Dept. of Computer Science, University of Illinois, Urbana, IL 61801.
Supported in part by NSF grants CCF-0728782
and CCF-1016684.
{\tt chekuri@cs.illinois.edu}}
\and
Alina Ene
\thanks{
Dept. of Computer Science, University of Illinois, Urbana, IL 61801.
Supported in part by NSF grants CCF-0728782 and CCF-1016684.
{\tt ene1@illinois.edu}}
}

\date{\today}

\maketitle

\thispagestyle{empty}
\begin{abstract}
  We study algorithms for the {\sc Submodular Multiway Partition}
  problem (\SubMP). An instance of \SubMP consists of a finite ground
  set $V$, a subset of $k$ elements $S = \{s_1,s_2,\ldots,s_k\}$
  called terminals, and a non-negative submodular set function
  $f:2^V\rightarrow \mathbb{R}_+$ on $V$ provided as a value oracle.
  The goal is to partition $V$ into $k$ sets $A_1,\ldots,A_k$ such
  that for $1 \le i \le k$, $s_i \in A_i$ and $\sum_{i=1}^k f(A_i)$ is
  minimized. \SubMP generalizes some well-known problems such as the
  {\sc Multiway Cut} problem in graphs and hypergraphs, and the {\sc
    Node-weighed Multiway Cut} problem in graphs. \SubMP for arbitrary
  submodular functions (instead of just symmetric functions) was
  considered by Zhao, Nagamochi and Ibaraki \cite{ZhaoNI05}. Previous
  algorithms were based on greedy splitting and divide and
  conquer strategies. In very recent work \cite{ChekuriE11} we
  proposed a convex-programming relaxation for \SubMP based on the
  Lov\'asz-extension of a submodular function and showed its
  applicability for some special cases. In this paper we obtain the
  following results for arbitrary submodular functions via this
  relaxation.
  \begin{itemize}
  \item A $2$-approximation for \SubMP. This improves the
    $(k-1)$-approximation from \cite{ZhaoNI05}. 
  \item A $(1.5-1/k)$-approximation for \SubMP when $f$ is {\em
    symmetric}. This improves the $2(1-1/k)$-approximation
    from \cite{Queyranne99,ZhaoNI05}.
  \end{itemize}
\end{abstract}

\newpage
\section{Introduction}
In this paper we consider the approximability of the following
problem.  

\medskip 
\noindent {\sc Submodular Multiway Partition}
(\SubMP). Let $f:2^V \rightarrow \mathbb{R}_+$ be a non-negative
submodular set function\footnote{A set function $f:2^V \rightarrow
  \mathbb{R}$ is submodular iff $f(A) + f(B) \ge f(A \cap B) + f(A
  \cup B)$ for all $A,B \subseteq V$. Moreover, $f$ is symmetric if $f(A) =
  f(V-A)$ for all $A \subseteq V$.} over $V$ and let $S =
\{s_1,s_2,\ldots,s_k\}$ be a set of $k$ terminals from $V$. The
submodular multiway partition problem is to find a partition of $V$
into $A_1,\ldots,A_k$ such that $s_i \in A_i$ and $\sum_{i=1}^k
f(A_i)$ is minimized.  An important special case is when $f$ is
symmetric and we refer to it as \SymSubMP.

\noindent
\textbf{Motivation and Related Problems:} We are motivated to consider
\SubMP for two reasons. First, \SubMP generalizes several problems
that have been well-studied. We discuss them now.  Perhaps the most
well-known of the special cases is the {\sc Multiway Cut} problem in
graphs (\MC): the input is an undirected edge-weighted graph $G=(V,E)$
and the goal is remove a minimum weight set of edges to separate a
given set of $k$ terminals \cite{DahlhausJPSY92}. Although the goal is
to remove edges, one can see this as a partition problem, and in fact
as a special case of \SymSubMP with the cut-capacity function as $f$.
\MC is NP-hard and APX-hard to approximate even for $k=3$
\cite{DahlhausJPSY92}.  One obtains two interesting and related
problems if one generalizes \MC to \emph{hypergraphs}. Let
$G=(V,\script{E})$ be an edge-weighted hypergraph.

\HMC is the problem where the goal is to remove a minimum-weight set
of hyperedges to disconnect the given set of terminals.  {\sc
  Hypergraph Multiway Partition} problem (\HMP) is the special case of
\SymSubMP where $f$ is the hypergraph-cut function: $f(A) = \sum_{e
  \in \delta(A)} w(e)$ where $w(e)$ is the weight of $e$ and
$\delta(A)$ is the set of all hyperedges that intersect $A$ but are
not contained in $A$. The distinction between \HMC and \HMP is that in
the former a hyperedge incurs a cost only once if the vertices in it
are split across terminals while in \HMP the cost paid by a hyperedge
is the number of non-trivial pieces it is partitioned into. Both
problems have several applications, in particular for circuit
partitioning problems in VLSI design \cite{AlpertK95}. We wish to draw
special attention to \HMC since it is approximation equivalent to the
{\sc Node-weighted Multiway Cut} problem in graphs (\nodeMC) where the
nodes have weights and the goal is to remove a minimum-weight subset
of nodes to disconnect a given set of terminals
\cite{GargVY94,GargVY04}.  An important motivation to consider \SubMP
is that \HMC can be cast as a special case of it \cite{ZhaoNI05}; the
reduction is simple, yet interesting, and we stress that the resulting
function $f$ is {\em not} necessarily symmetric.  Since \nodeMC is
approximation-equivalent to \HMC and \HMC is a special case of \SubMP
it follows that one can view \nodeMC indirectly as a partition problem
with an appropriate submodular function. We believe this is a useful
observation that should be more widely-known. In fact, \SubMP (and
related generalizations) were introduced by Zhao, Nagamochi and
Ibaraki \cite{ZhaoNI05} partly motivated by the applications to
hypergraph cut and partition problems.

A second important motivation to consider \SubMP and \SymSubMP is the
following question. To what extent do current algorithms and
techniques for important special cases such as \MC and \nodeMC depend
on the fact that the underlying structure is a graph (or a
hypergraph)? Or is it the case that submodularity of the cut function
the key underlying phenomenon?  For \MC Dahlhaus \etal
\cite{DahlhausJPSY92} gave a simple $2(1-1/k)$-approximation via the
isolating cut heuristic. Queyranne \cite{Queyranne99} showed that this
same bound can be achieved for \SymSubMP (see also \cite{ZhaoNI05}).
For \MC, Calinescu, Karloff and Rabani \cite{CalinescuKR98}, in a
breakthrough, obtained a $1.5-1/k$ approximation via an interesting
geometric relaxation. The integrality gap for this relaxation has been
subsequently improved to $1.3438$ by Karger \etal \cite{KargerKSTY99}.
Once again it is natural to ask if this geometric relaxation is
specific to graphs and whether corresponding results exist for
\SymSubMP. Further, the current best approximation for \SubMP is $(k-1)$
\cite{ZhaoNI05} and is obtained via a simple greedy splitting
algorithm. \SubMP generalizes \nodeMC and the latter has a $2(1-1/k)$
approximation \cite{GargVY04} but it is a non-trivial LP
relaxation based algorithm. Therefore it is reasonable to expect that
one needs a mathematical programming relaxation to obtain a constant
factor approximation for \SubMP.

In very recent work \cite{ChekuriE11} we developed a simple and
straightforward convex-programming relaxation for \SubMP via the
Lov\'{a}sz-extension of a submodular function (we discuss this in more
detail below). An interesting observation is that this specializes to
the CKR-relaxation when we consider \MC! A natural question, that was
raised in \cite{ChekuriE11}, is whether the convex relaxation can be
used to obtain a better than $2$ approximation for \SymSubMP. In this
paper we answer this question in the positive and also obtain a
$2$-approximation for \SubMP improving the known $(k-1)$-approximation
\cite{ZhaoNI05}. We now describe the convex relaxation.

\medskip \noindent \textbf{A convex relaxation via the Lov\'{a}sz
  extension}.  The relaxation \SubMPRel for \SubMP was
introduced in \cite{ChekuriE11}. It is based on the well-known
Lov\'asz extension of a submodular function which we describe for
completeness.  Let $V$ be a finite ground set of size $n$, and let $f:
2^V \rightarrow \mathbb{R}$ be a real-valued set function. By
representing a set by its characteristic vector, we can think of $f$
as a function that assigns a value to each vertex of the boolean
hypercube $\{0, 1\}^n$. The Lov\'{a}sz extension $\hat{f}$ extends $f$
to all of $[0, 1]^n$ and is defined as follows\footnote{The standard
  definition is slightly different but is equivalent to the one we
  give; see \cite{Vondrak09}, \cite{ChekuriE11}.}:
	$$\hat{f}(\bx) = \Ex_{\theta \in
	[0,1]}\left[f(\bx^{\theta})\right] =
	\int_0^1 f(\bx^{\theta}) d\theta$$
where $\bx^{\theta} \in \{0, 1\}^n$ for a given vector $\bx \in
[0,1]^n$ is defined as: $x_i^{\theta} = 1$ if $x_i \geq \theta$ and
$0$ otherwise.
Lov\'asz showed that $\hat{f}$ is convex if and only if $f$ is
submodular \cite{Lovasz83}.

\SubMP is defined as a partition problem. We can equivalently
interpret it as an {\em allocation} problem (also a labeling problem)
where for each $v \in V$ we decide which of the $k$ terminals it is
allocated to. Thus we have non-negative variables $x(v,i)$ for $v \in
V$ and $1 \le i \le k$ and an allocation is implied by the simple
constraint $\sum_{i=1}^k x(v,i) = 1$ for each $v$. Of course a
terminal $s_i$ is allocated to itself. The only complexity is in the
objective function since $f$ is submodular. However, the Lov\'{a}sz
extension gives a direct and simple way to express the objective
function.  Let $\bx_i$ be the vector obtained by restricting
$\bx$ to the $i$'th terminal $s_i$; that is $\bx_i =
(x(v_1,i),\ldots,x(v_n,i))$. If $\bx_i$ is integral
$\hat{f}(\bx_i) = f(A_i)$ where $A_i$ is the support of $\bx_i$.
Therefore we obtain the following relaxation.

\begin{center}
\begin{boxedminipage}{0.5\linewidth}
\vspace{-0.2in}
\begin{align*}
& \textbf{\SubMPRel}\\
\min \quad & \sum_{i = 1}^k \hat{f}(\bx_i)\\
& \sum_{i = 1}^k x(v, i)  &= \qquad 1 & \qquad \forall v\\
& x(s_i, i)  &= \qquad 1 & \qquad \forall i\\
& x(v, i)  &\ge \qquad 0 & \qquad \forall v, i
\end{align*}
\end{boxedminipage}
\end{center}

The above relaxation can be solved in polynomial time\footnote{The
running time is polynomial in $n$ and $\log\left(\max_{S \subseteq V}
f(S)\right)$.} via the ellipsoid method \cite{ChekuriE11}. We give
algorithms to round an optimum fractional solution to \SubMPRel and
obtain the following two results which also establish corresponding
upper bounds on the integrality gap of \SubMPRel.

\begin{theorem} \label{thm:smp-symmetric}
  There is a $(1.5-1/k)$-approximation for \SymSubMP.
\end{theorem}

\begin{theorem} \label{thm:smp-general}
  There is a $2$-approximation for \SubMP.
\end{theorem}

\begin{remark} {\em It is shown in \cite{GargVY04} that an
    $\alpha$-approximation for \nodeMC implies an
    $\alpha$-approximation for the {\sc Vertex Cover}
    problem. Therefore, improving the $2$-approximation for \SubMP is
    infeasible without a corresponding improvement for {\sc Vertex
      Cover}. It is easy to show that the integrality gap of \SubMPRel
    is is at least $2(1-1/k)$ even for instances of \HMC. The best
    lower bound on the integrality gap of \SubMPRel for \SymSubMP that
    we know is $8/(7 + \frac{1}{k-1})$, the same as that for \MC shown
    in \cite{FreundK00}.}
\end{remark}

\begin{remark}
  {\em Related to \SubMP and \SymSubMP are $k$-way partition problems
    \kSubMP and \kSymSubMP where no terminals are specified but the
    goal is to partition $V$ into $k$ non-empty sets $A_1,\ldots, A_k$
    to minimize $\sum_{i=1}^k f(A_i)$.  When $k$ is part of the input
    these problems are NP-Hard but for fixed $k$, \kSymSubMP admits a
    polynomial time algorithm \cite{Queyranne99} while the status of
    \kSubMP is still open. For fixed $k$ one can reduce \kSubMP to
    \SubMP by guessing $k$ terminals and this leads to a
    $2$-approximation via Theorem~\ref{thm:smp-general}, improving the
    previously known ratio of $(k+1-2\sqrt{k-1})$ \cite{OkumotoFN10}.}
\end{remark}

Our results build on some basic insights that were outlined in
\cite{ChekuriE11} where the special cases of \HMP and \HMC were
considered (among other results). In \cite{ChekuriE11} a
$(1.5-1/k)$-approximation for \HMP and a $\min\{2(1-1/k),
H_{\Delta}\}$-approximation for \HMC were given where $\Delta$ is the
maximum hyperedge degree and $H_i$ is the $i$'th harmonic number. Our
contribution in this paper is a non-trivial, and technical new result
on rounding \SubMPRel (Theorem~\ref{thm:smp-main} below) that applies
to an arbitrary submodular function. The {\em formulation} of the
statement of the theorem may appear natural in retrospect but was a
significant part of the difficulty. We now give an overview of the
rounding algorithm(s) and the new result. We then discuss and compare
to prior work.

\subsection{Overview of rounding algorithms and the main technical result}
\label{subsec:overview}

Let $\bx$ be a fractional allocation and $\sum_i \hat{f}(\bx_i)$ the
corresponding objective function value. How do we round $\bx$ to an
integral allocation while approximately preserving the convex
objective function? The simple insight in \cite{ChekuriE11} is that we
simply follow the definition of the Lov\'{a}sz function and do {\em
$\theta$-rounding}: pick a (random) threshold $\theta \in [0, 1]$ and
set $x(v_j,i) = 1$ if and only if $x(v_j,i) \ge \theta$. Let
$\tbx(\theta)$ be the resulting integer vector. If we pick $\theta$
uniformly at random in $[0,1]$ then the expected cost of $\sum_i
\Ex[f(\tbx_i(\theta))] = \sum_i \hat{f}(\bx_i)$. However, the problem
is that $\tbx(\theta)$ may not correspond to a feasible allocation.
Let $A(i,\theta)$ be the support of $\tbx_i$, that is, the set of
vertices assigned to $s_i$ for a given $\theta$. The reason that
$\tbx(\theta)$ may not be a feasible allocation is two-fold. First, a
vertex $v$ may be assigned to multiple terminals, that is, the sets
$A(i,\theta)$ for $i=1,\ldots,k$ may not be disjoint. Second, the
vertices $U(\theta) = V - \cup_{i=1}^k A(i,\theta)$ are unallocated.
We let $A(\theta) = \cup_{i=1}^k A(i,\theta)$ be the allocated set.

Our fundamental insight here is that the expected cost of the
unallocated set, that is $f(U(\theta))$, can be upper bounded
effectively. We can then assign the set $U(\theta)$ to an arbitrary
terminal and use sub-additivity of $f$ (since it is submodular and
non-negative).  Before we formalize this, we discuss how to overcome
the overlap in the sets $A(i,\theta)$. If $f$ is symmetric then it is
also posi-modular and one can do a simple uncrossing of the sets to
make them disjoint without increasing the cost. If $f$ is not
symmetric we cannot resort to this trick; in this case we ensure that
the sets $A(i,\theta)$ are disjoint by picking $\theta$ uniformly in
$(1/2,1]$ rather than $[0,1]$ (we call this half-rounding). Now the
unallocated set and the expected cost of the initial allocation are
some what more complex. We analyze both these scenarios using the
following theorem which is our main result. The theorem below has a
parameter $\delta \in (1/2,1]$ and this corresponds to rounding where
we pick $\theta$ uniformly from the interval $(1-\delta, 1]$.

\begin{theorem} \label{thm:smp-main} Let $\bx$ be a feasible solution
  to \SubMPRel. For $\theta \in [0,1]$ let $A(i,\theta) = \{ v \mid
  x(v,i) \ge \theta \}$, $A(\theta) = \cup_{i=1}^k A(i,\theta)$ and
  $U(\theta) = V - A(\theta)$.  For any $\delta \in [1/2, 1]$, we have
	$$\sum_{i = 1}^k \int_0^{\delta} f(A(i, \theta) d\theta \geq
	\int_0^{\delta} f(A(\theta))d\theta + \int_0^1
	f(U(\theta))d\theta.$$
\end{theorem}

\noindent
By setting $\delta = 1$, we get the following corollary.

\begin{corollary} \label{cor:smp-main-simple}
	$$\sum_{i=1}^k \hat{f}(\bx_i) = \sum_{i = 1}^k \int_0^{1} f(A(i, \theta) d\theta \geq \int_0^1 f(A(\theta)) d\theta + \int_0^1
	f(U(\theta)) d\theta.$$
\end{corollary}

\noindent
Theorem~\ref{thm:smp-main} gives a unified analysis of our
algorithms for \SymSubMP and \SubMP. More precisely, we get
Theorem~\ref{thm:smp-symmetric} and Theorem~\ref{thm:smp-general} as
rather simple corollaries. Corollary~\ref{cor:smp-main-simple} is
sufficient to show that the \SymSubMP algorithm achieves a
$1.5$-approximation and that the \SubMP algorithm achieves a
$4$-approximation. In order to show that the \SubMP algorithm achieves
a $2$-approximation, we need the stronger statement of
Theorem~\ref{thm:smp-main}.

\subsection{Discussion and other related work}
\label{subsec:related-work}
Our recent work \cite{ChekuriE11} considered the \MCSAfull (\MCSA);
\SubMP is a special case. \MCSA also contains as special cases other
problems such as uniform metric labeling, non-metric facility
location, hub location and variants. The main insight in
\cite{ChekuriE11} is that a convex programming relaxation via the
\lovasz extension follows naturally for \MCSA and hence
$\theta$-rounding based algorithms provide a unified way to understand
and extend several previous results. The integrality gap of \SubMPRel for
\SymSubMP and \SubMP were posed as open questions following results for
the special cases of \HMC and \HMP. These results subsequently 
inspired the formulation of Theorem~\ref{thm:smp-main}.

Geometry plays a key role in the formulation, rounding and analysis of
the relaxation proposed for \MC by Calinescu, Karloff and Rabani
\cite{CalinescuKR98}; they obtained a $1.5-1/k$ approximation.  The
subsequent work of Karger \etal exploits the geometric aspects further
to obtain an improvement in the ratio to $1.3438$. If one views \MC as
a special case of \SymSubMP then the function $f$ under consideration
is the cut function. The cut function $f$ can be decomposed into
several simple submodular functions, corresponding to the edges, each
of which depends only on two vertices. This allows one to focus on the
probability that an edge is cut in the rounding process. Our work in
\cite{ChekuriE11} for \HMP and \HMC is also in a similar vein since
one can visualize and analyze the simple functions that arise from the
hypergraph cut function.  Our current analysis differs substantially
in that we no longer have a local handle on $f$, and hence the need
for Theorem~\ref{thm:smp-main}. It is interesting that the integrality
gap of \SubMPRel is at most $1.5-1/k$ for any symmetric function $f$,
matching the bound achieved by \cite{CalinescuKR98} for \MC. Our
rounding differs from that in \cite{CalinescuKR98}; both do
$\theta$-rounding but our algorithm uncrosses the sets $A(i,\theta)$
to make them disjoint while CKR-rounding does it by picking a random
permutation. One can understand the random permutation as an oblivious
uncrossing operation that is particularly suited for submodular
functions that depend on only two variables (in this case the edges); it is 
unclear whether this is suitable for arbitrary symmetric
functions.

As we remarked, \SymSubMP and \SubMP were considered in several papers
\cite{Queyranne99,ZhaoNI05,OkumotoFN10} with \HMC and \kHMC as
interesting applications for \SubMP. These papers primarily relied on
greedy methods.  It was noted in \cite{ZhaoNI05} that \HMC and \nodeMC
are essentially equivalent problems.  Garg, Vazirani and Yannakakis
\cite{GargVY04} gave a $2(1-1/k)$-approximation for \nodeMC
\cite{GargVY04} via a natural distance based LP relaxation; we note
that this result is non-trivial and relies on proving the existence of
a half-integral optimum fractional solution. Viewing \nodeMC as
equivalent to \HMC allows one to reduce it to \SubMP, and as we noted
in \cite{ChekuriE11} \SubMPRel gives a new and strictly stronger
relaxation for \nodeMC.  The previous best approximation for \SubMP
was $(k-1)$ \cite{ZhaoNI05}.  As we already remarked, 
obtaining a constant factor approximation for \SubMP without 
a mathematical programming relaxation like \SubMPRel is difficult
given the lack of combinatorial algorithms for special cases
like \nodeMC.

Submodular functions play a fundamental role in classical
combinatorial optimization. In recent years there have been
several new results on approximation algorithms for 
problems with objective functions that depend on submodular functions.
In addition to combinatorial techniques such as greedy and
local-search, mathematical programming methods have been particularly
important. It is natural to use the \lovasz extension for problems
involving minimization since the extension is convex; see
\cite{IwataN09,GoelKTW09,ChekuriE11} for instance. For maximization
problems involving submodular functions the multilinear extension
introduced in \cite{CalinescuCPV07} has been  useful
\cite{Vondrak08,KulikST09,LeeMNS10,Vondrak09,ChekuriVZ11}.

\section{Symmetric Submodular Multiway Partition}
We consider the following algorithm to round a feasible solution $\bx$
to \SubMPRel.
\begin{algo}
\underline{\textbf{\SymMPR}}
\\\> let $\bx$ be a feasible solution to \SubMPRel
\\\> pick $\theta \in [0, 1]$ uniformly at random
\\\> for $i = 1$ to $k$
\\\>\> $A(i, \theta) \leftarrow \{v \;|\; x(v, i) \geq \theta\}$
\\\> $A(\theta) \leftarrow \bigcup_{1 \leq i \leq k} A(i,
\theta)$
\\\> $U(\theta) \leftarrow V - A(\theta)$
\\\> for $i = 1$ to $k$
\\\>\> $A'_i \leftarrow A(i, \theta)$
\\\> \Comment{uncross $A'_1, \cdots, A'_k$}
\\\> while there exist $i \neq j$ such that $A'_i \cap A'_j \neq
\emptyset$
\\\>\> if $\left(f(A'_i) + f(A'_j - A'_i) \leq f(A'_i) + f(A'_j)\right)$
\\\>\>\> $A'_j \leftarrow A'_j - A'_i$
\\\>\> else
\\\>\>\> $A'_i \leftarrow A'_i - A'_j$
\\\> return $(A'_1, \cdots, A'_{k - 1}, A'_k \cup U(\theta))$
\end{algo}

\noindent
We prove the following theorem.

\begin{theorem}
  \label{thm:sym}
  Let $\bx$ be a feasible solution to \SubMPRel. If $f$ is
  a symmetric submodular function, the algorithm \textbf{\SymMPR} outputs a
  valid multiway partition of expected cost at most $1.5 \cdot
  \sum_{i=1}^k \hat{f}(\bx_i)$.
\end{theorem}

\noindent
The algorithm does $\theta$-rounding in the interval $[0,1]$ to obtain
(random) sets $A(i,\theta)$ for $i=1,\ldots,k$. Let $\optcr = \sum_{i
  = 1}^k \hat{f}(\bx_i)$. Note that $\Ex[f(A(i,\theta))] =
\hat{f}(\bx_i)$ and hence $\sum_{i=1}^k \Ex[f(A(i,\theta))] =
\sum_{i=1}^k \hat{f}(\bx_i) = \optcr$. The lemma below shows that
the uncrossing operation does not increase the cost. This is
was used in the context of
multiway cuts previously \cite{SvitkinaT04,ChekuriE11}; we
include the proof for completeness.

\begin{lemma}[\cite{ChekuriE11}] \label{lem:uncrossing}
	Let $A'_1, \ldots, A'_k$ denote the sets after
        uncrossing the sets $A(1,\theta),\ldots,A(k,\theta)$.  If $f$
        is a symmetric submodular function then $\cup_{i=1}^k A'_i =
        \cup_{i=1}^k A(i,\theta)$ and
		$$\sum_{i = 1}^k f(A'_i) \leq
		\sum_{i = 1}^k f(A(i, \theta)).$$ 
\end{lemma}
\begin{proof}
  In each uncrossing step we replace $A'_i$ and $A'_j$ either by
  $A'_i$ and $A'_j - A'_i$ or by $A'_i-A'_j$ and $A'_j$. Since $f$ is
  submodular and symmetric, $f$ is posi-modular; that is, for any two
  sets $X$ and $Y$, $f(X) + f(Y) \geq f(X - Y) + f(Y - X)$. Therefore,
  for any two sets $X$ and $Y$, $\min\{f(X - Y) + f(Y), f(X) + f(Y -
  X)\}$ is at most $f(X) + f(Y)$. Thus it follows by induction that
  $\sum_{i = 1}^k f(A'_i) \leq \sum_{i = 1}^k f(A(i, \theta))$ and
  $\cup_{i=1}^k A'_i = \cup_{i=1}^k A(i,\theta)$.
\end{proof}

\begin{lemma} \label{cor:sym-smp-unallocated}
	If $f$ is a symmetric submodular function,
		$$\Ex_{\theta \in [0, 1]}[f(U(\theta))] \leq {1 \over 2}
		\optcr.$$
\end{lemma}
\begin{proof}
	By setting $\delta = 1$ in Theorem~\ref{thm:smp-main}, we get
		$$\optcr \geq \int_0^1 f(V - U(\theta)) d\theta + \int_0^1
		f(U(\theta)) d\theta.$$
	Since $f$ is symmetric, $f(V-U(\theta)) = f(U(\theta)$ for all
        $\theta$ and hence,
		$$\optcr \geq 2 \int_0^1
		f(U(\theta)) d\theta = 2 \Ex_{\theta \in [0, 1]}[f(U(\theta))].$$
\end{proof}

\noindent
The random partition returned by the algorithm is
$(A'_1,\ldots,A'_{k-1},A'_k \cup U(\theta))$. A non-negative
submodular function is sub-additive, hence $f(A'_k \cup U(\theta)) \le
f(A'_k) + f(U(\theta))$. The expected cost of the partition is
\begin{align*}
	\sum_{i=1}^{k-1} \Ex[f(A'_i)] + \Ex[f(A'_k \cup U(\theta))] & \le
	\sum_{i=1}^k \Ex[f(A'_i)] + \Ex[f(U(\theta))] \\
	& \le \sum_{i=1}^k \Ex[f(A(i,\theta))] + \Ex[f(U(\theta))] \quad
	(\mbox{Using Lemma~\ref{lem:uncrossing}})\\
	& \le \optcr + \frac{1}{2} \optcr \quad (\mbox{Using
	Lemma~\ref{cor:sym-smp-unallocated}})\\
   	& = 1.5 \optcr.
\end{align*}
This finishes the proof of Theorem~\ref{thm:sym}. It is not hard to 
verify that the algorithm runs in polynomial time. One can easily
derandomize the algorithm as follows. The only randomness is in the
choice of $\theta$.  As $\theta$ ranges in the interval $[0,1]$, the
collection of sets $\{A(i,\theta) \mid 1 \le i \le k\}$ changes only
when $\theta$ crosses some $x(v_j,i)$ value. Thus there are at most $nk$
such distinct values. We can try each of them as a choice for $\theta$ 
and pick the least cost partition obtained among all the choices.

\medskip
\noindent
\textbf{Achieving a $(1.5-1/k)$-aproximation:} We can improve the
approximation to $1.5 - 1/k$ as follows. We relabel the terminals so
that $k = \argmax_{1 \leq i \leq k} \hat{f}(\bx_i)$.  We perform
$\theta$-rounding with respect to the first $k - 1$ terminals in order
to get the sets $A(i, \theta)$ for each $i \neq k$, and we let
$U(\theta) = V - \cup_{1 \leq i \leq k - 1} A(i, \theta)$.  We uncross
the sets $\{A(i, \theta) \mid 1\le i < k\}$ to get $k - 1$ disjoint
sets $A'_i$, and we return $(A'_1, \cdots, A'_{k - 1}, U(\theta))$. We
can prove a variant of Theorem~\ref{thm:smp-main} that shows that the
expected cost of $U(\theta)$ is at most $\optcr / 2$, even when
$U(\theta)$ is the set of all vertices that are unallocated when we
perform $\theta$-rounding with respect to only the first $k'$
terminals, for any $k' \leq k$. The proof of this extension of
Theorem~\ref{thm:smp-main} is notationally and technically messy (and
somewhat non-trivial), and we omit it in this version of the
paper. The total expected cost of the sets $A'_1, \cdots, A'_{k - 1}$
is at most $(1 - 1/k) \optcr$ (since we saved on $\hat{f}(\bx_k)$),
and the expected cost of $U(\theta)$ is at most $\optcr / 2$.

\section{Submodular Multiway Partition}
In this section we consider \SubMP when $f$ is an arbitrary non-negative
submodular function. We choose $\theta \in (1/2,1]$ to ensure that
the sets $\{A(i,\theta) \mid 1 \le i \le k\}$ are disjoint.
\begin{algo}
\underline{\textbf{\SubMPRH}}
\\\> let $\bx$ be a feasible solution to \SubMPRel
\\\> pick $\theta \in (1/2, 1]$ uniformly at random
\\\> for $i = 1$ to $k$
\\\>\> $A(i, \theta) \leftarrow \{v \;|\; x(v, i) \geq \theta\}$
\\\> $A(\theta) \leftarrow \bigcup_{1 \leq i \leq k} A(i,
\theta)$
\\\> $U(\theta) \leftarrow V - A(\theta)$
\\\> return $(A(1, \theta), \cdots, A(k-1, \theta), A(k,\theta) \cup
U(\theta))$
\end{algo}

\begin{proofof}{Theorem~\ref{thm:smp-general}}
	In the following, we will show that \SubMPRH achieves a
	$2$-approximation for \SubMP. As before, let $\optcr = \sum_{i =
	1}^k \hat{f}(\bx_i)$. Since $f$ is subadditive, the expected cost
	of the partition returned by \SubMPRH is
	\begin{align*}
		\Ex_{\theta \in (1/2, 1]}\Bigg[\sum_{i = 1}^{k - 1} f(A(i,
		\theta)) + f(A(k,\theta) \cup U(\theta))\Bigg] 
                  & \le \Ex_{\theta \in (1/2, 1]}\Bigg[\sum_{i = 1}^{k} f(A(i,
		\theta)) + f(U(\theta))\Bigg] \\
		&= 2 \left( \sum_{i = 1}^{k} \int_{1/2}^1 f(A(i, \theta))
		d\theta + \int_{1/2}^1 f(U(\theta)) d\theta\right)\\
		&= 2\left( \optcr 
		- \sum_{i = 1}^{k} \int_0^{1/2} f(A(i, \theta)) d\theta
		+ \int_{1/2}^1 f(U(\theta)) d\theta\right).
	\end{align*}
        To show that the expected cost is at most $2\optcr$ it
        suffices to show that $\sum_{i = 1}^{k} \int_0^{1/2} f(A(i,
        \theta)) d\theta \ge \int_{1/2}^1 f(U(\theta))
        d\theta$. Setting $\delta = 1/2$ in
        Theorem~\ref{thm:smp-main} we get
	\begin{align*}
		\sum_{i = 1}^{k} \int_0^{1/2} f(A(i, \theta)) d\theta  & \geq
		\int_0^{1/2} f(V - U(\theta)) d\theta + \int_0^1 f(U(\theta))
		d\theta\\
		&\geq \int_{1/2}^1 f(U(\theta)) d\theta \qquad\qquad
		\mbox{($f$ is non-negative)}
	\end{align*}
	Thus \SubMPRH achieves a randomized $2$-approximation for
        \SubMP.  The algorithm can be derandomized in the same fashion
        as the one for symmetric functions since there are at most $nk$
       values of $\theta$ (the $x(v_j,i)$ values) where the partition
       returned by the algorithm can change.
\end{proofof}

\noindent
\textbf{Improving the factor of $2$:} As we remarked earlier the {\sc
Vertex Cover} problem can be reduced in an approximation preserving
fashion to \SubMP, and hence it is unlikely that the factor of $2$ for
\SubMP can be improved. However, it may be possible to obtain a
$2(1-1/k)$-approximation.  A natural algorithm here is to do
half-rounding only with respect to the first $k-1$ terminals, where $k
= \argmax_i \hat{f}(\bx_i)$, and assign all the remaining elements to
$k$. We have so far been unable to strengthen
Theorem~\ref{thm:smp-main} to achieve the desired improvement.

\section{Proof of Main Theorem}
In this section we prove Theorem~\ref{thm:smp-main}, our main
technical result.  We recall some relevant definitions.  Let $\bx$ be
a solution to \SubMPRel. We are interested in analyzing
$\theta$-rounding when $\theta$ is chosen uniformly at random from an
interval $[1-\delta,1]$ for some $\delta \ge 0$. For a label $i$ let
$A(i, \theta) = \{v \in V \sp x(v, i) \geq \theta\}$ be the set of all
vertices that are assigned/allocated to $i$ for some fixed $\theta$.
Note that for distinct labels $i,i'$ the sets $A(i,\theta)$ and
$A(i',\theta)$ may not be disjoint if $\theta \le 1/2$, although they
are disjoint if $\theta > 1/2$.  Let $A(\theta) = \bigcup_{1 \leq
i \leq k} A(i, \theta)$ be the set of all vertices that are allocated
to the terminals when $\theta$ is the chosen threshold. We let
$U(\theta) = V - A(\theta)$ denote the set of unallocated vertices.
With this notation in place we restate Theorem~\ref{thm:smp-main}.

\begin{theorem} 
	For any $\delta \in [1/2, 1]$, we have
	$$\sum_{i = 1}^k \int_0^{\delta} f(A(i, \theta) d\theta \geq
	(k\delta - \delta - 1) f(\emptyset) + \int_0^{\delta}
	f(A(\theta))d\theta + \int_0^1 f(U(\theta))d\theta .$$
\end{theorem}

The proof of the above theorem is somewhat long and technical. At a
high-level it is based on induction on the number of vertices with a
particular ordering that we discuss now. In the following, we use $i$
to index over the labels, and we use $j$ to index over the vertices.
For vertex $v_j$, let $\alpha_j = \max_i x(v_j,i)$ be the maximum
amount to which $\bx$ assigns $v_j$ to a label. We relabel the
vertices such that $0 \le \alpha_1 \le \alpha_2 \ldots \le \alpha_n
\le 1$.  For notational convenience we let $\alpha_0 = 0$ and
$\alpha_{n+1} = 1$.  Further, for each vertex $v_j$ we let $\ell_j$ be
a label such that $\alpha_j = x(v_j,\ell_j)$; note that $\ell_j$ is
not necessarily unique unless $\alpha_j > 1/2$.\footnote{An alert
reader may notice that we do not distinguish between terminals and
non-terminals. In fact the theorem statement does not rely on the fact
that terminals are assigned fully to their respective labels. The only
place we use the fact that $x(s_i, i) = 1$ for each $i$ is to show
that $\theta$-rounding based algorithms produce a valid multiway
partition with respect to the terminals.}

We observe that in $\theta$-rounding, $v_j$ is allocated to a terminal
(that is $v_j \in A(\theta)$) if and only if $\theta \le \alpha_j$,
otherwise $v_j \in U(\theta)$ and thus it is unallocated. It follows
from our ordering that $U(\theta) = \{v_1,v_2, \ldots, v_{j-1}\}$ iff
$\theta \in (\alpha_{j-1},\alpha_j]$ and in this case $A(\theta) = V-
U(\theta) = \{v_j,\ldots,v_n\}$. Thus, prefixes of the ordering given
by the $\alpha$ values are the only interesting sets to consider when
analyzing the rounding process from the point of view of allocated and
unallocated vertices. To help with notation, for $1 \le j \le n$ we
let $V_j = \{v_1,v_2,\ldots,v_j\}$ and $V_0 = \emptyset$.  The
following proposition captures this discussion.

\begin{prop} \label{prop:allocated-helper}
	Let $\alpha_0 = 0$ and let $j$ be any index such that $1 \leq j
	\leq n$. For any $\theta \in (\alpha_{j - 1}, \alpha_j]$,
	$A(\theta) = V - V_{j - 1}$ and $U(\theta) = V_{j - 1}$.
\end{prop}

\noindent
It helps to rewrite the expected cost of $f(A(\theta))$ and
$f(U(\theta))$ under $\theta$-rounding in a more convenient form given
below.

\begin{prop} \label{prop:allocated}
	Let $r \in [0, 1]$, and let $h$ be the largest value of $j$ such
	that $\alpha_j \leq r$. We have
		$$\int_0^r f(A(\theta)) d\theta = \sum_{j = 1}^h
		\alpha_j (f(V - V_{j - 1}) - f(V - V_j)) + r f(V - V_h)$$
	and
		$$\int_0^r f(U(\theta)) d\theta = \sum_{j = 1}^h
		\alpha_j (f(V_{j - 1}) - f(V_j)) + r f(V_h).$$
\end{prop}
\begin{proof}
   Recall from Proposition~\ref{prop:allocated-helper} 
   that $A(\theta) = V-V_{j-1}$ when $\theta \in (\alpha_{j-1},\alpha_j]$.
   Therefore,
   \begin{eqnarray*}
     \int_0^r f(A(\theta)) d\theta &=& \sum_{j = 1}^{h}
     \int_{\alpha_{j - 1}}^{\alpha_j} f(A(\theta)) d\theta + \int_{\alpha_h}^r f(A(\theta)d\theta\\
     &=& \sum_{j = 1}^{h} (\alpha_j - \alpha_{j - 1}) f(V -
     V_{j - 1}) + (r-\alpha_h) f(V - V_h) \\
     &=&  r f(V - V_h) + \sum_{j = 1}^h \alpha_j (f(V - V_{j - 1}) - f(V - V_j)).
   \end{eqnarray*}
   The second identity follows from a very similar argument.
\end{proof}

\paragraph{The inductive approach:} Recall that numbering the vertices
in increasing order of their $\alpha$ values ensures that $U(\theta)$
is $V_j$ for some $0 \le j \le n$.  Let $\bx_j$ be the restriction of
$\bx$ to $V_j$. Note that $\bx_j$ gives a feasible allocation of $V_j$
to the $k$ labels although it does not necessarily correspond to a
multiway partition with respect to the original terminals. Also, note
that the function $f$ when restricted to $V_j$ is still submodular but
may not be symmetric even if $f$ is. In order to argue about $\bx_j$
we introduce additional notation. Let $A_j(i, \theta) = A(i, \theta)
\cap V_j$, $A_j(\theta) = A(\theta) \cap V_j$, and $U_j(\theta) =
U(\theta) \cap V_j$. In other words $A_j(\theta)$ and $U_j(\theta)$
are the allocated and unallocated sets if we did $\theta$-rounding
with respect to $\bx_j$ that is defined over $V_j$.

Let $\rho_j = \sum_{i = 1}^k \int_0^{\delta} f(A_j(i,
\theta))d\theta$; we have $\rho_0 = k \delta f(\emptyset)$.  Note that
the left hand side of the inequality in Theorem~\ref{thm:smp-main} is
$\rho_n = \sum_{i = 1}^k \int_0^{\delta} f(A_n(i, \theta))d\theta$,
since $A_n(i,\theta) = A(i,\theta)$. To understand $\rho_n$ we
consider the quantity $\rho_j - \rho_{j-1}$ which is easier since
$\rho_j$ and $\rho_{j-1}$ differ only in $v_j$.  Recall that $\ell_j$
is a label such that $\alpha_j = \max_{i=1}^k x(v_j,i)$. The
importance of $\ell_j$ is that if $v_j$ is allocated then it is
allocated to $\ell_j$ (and possibly to other labels as well).  We
express $\rho_j - \rho_{j-1}$ as the sum of two quantities with the
term for $\ell_j$ separated out.

\begin{prop}
  \label{prop:inductive}
  $$\rho_j - \rho_{j-1} = \int_{0}^\delta \left(f(A_j(\ell_j,\theta))
  - f(A_{j-1}(\ell_j,\theta)\right)d\theta + \sum_{i \neq \ell_j}
  \int_{0}^\delta \left(f(A_j(i,\theta)) -
  f(A_{j-1}(i,\theta)\right)d\theta.$$
\end{prop}

\noindent
We prove the following two key lemmas by using of submodularity of $f$
appropriately.

\begin{lemma}
  \label{lem:charging1}
  For any $\delta \in [1/2,1]$ and for any $j$ such that $1 \le j \le n$,
  $$\sum_{i \neq \ell_j}  \int_{0}^\delta (f(A_j(i,\theta)) -
  f(A_{j-1}(i,\theta))d\theta \ge f(V_j) - f(V_{j-1}) + \alpha_j
  (f(V_{j-1}) - f(V_j)).$$
\end{lemma}

\noindent
Summing the left hand side in the above lemma over all $j$ and
applying Proposition~\ref{prop:allocated} with $r=1$ we obtain:
\begin{corollary}
  \label{cor:charging1}
 $$\sum_{j=1}^n \left(\sum_{i \neq \ell_j}  \int_{0}^\delta
 (f(A_j(i,\theta)) - f(A_{j-1}(i,\theta))d\theta\right) \ge \int_0^1
 f(U(\theta))d\theta  - f(\emptyset).$$
\end{corollary}

\medskip
\noindent
Our second key lemma below is the more involved one. Unlike the first
lemma above we do not have a clean and easy expression for a single
term $\int_{0}^\delta \left(f(A_j(\ell_j,\theta)) -
  f(A_{j-1}(\ell_j,\theta)\right)d\theta$ but the sum over all $j$
gives a nice telescoping sum that results in the bound below.

\begin{lemma}
  \label{lem:charging2}
  Let $\delta \in [0,1]$ and let $h$ be the largest value of $j$
  such that $\alpha_j \le \delta$.
  $$\sum_{j=1}^n  \left (\int_{0}^\delta (f(A_j(\ell_j,\theta)) -
  f(A_{j-1}(\ell_j,\theta))d\theta\right) \ge \int_0^\delta
  f(A(\theta)) d\theta -  \delta f(\emptyset).$$
\end{lemma}

\noindent
The proofs of the above lemmas are given in
Sections~\ref{sec:charging1} and \ref{sec:charging2} respectively. We
now finish the proof of Theorem~\ref{thm:smp-main} assuming the above
two lemmas.

\begin{proofof}{Theorem~\ref{thm:smp-main}}
  Let $h$ be the largest value of $j$ such that $\alpha_j \leq
  \delta$.  From Proposition~\ref{prop:inductive} we have
  \begin{align*}
	\sum_{i=1}^k & \int_0^\delta f(A(i,\theta))d\theta  = \rho_n 
	= \rho_0 + \sum_{j=1}^n (\rho_j - \rho_{j-1}) \\
	&= \rho_0 + \sum_{j=1}^n \left( \int_{0}^\delta
	(f(A_j(\ell_j,\theta)) - f(A_{j-1}(\ell_j,\theta))d\theta +
	\sum_{i \neq \ell_j}  \int_{0}^\delta (f(A_j(i,\theta)) -
	f(A_{j-1}(i,\theta))d\theta \right)\\
	& \hspace{5in}(\mbox{Use Proposition~\ref{prop:inductive}})\\
	&\ge \rho_0 +  \int_0^\delta f(A(\theta)) d\theta -  \delta
	f(\emptyset)  + \sum_{j=1}^n \left(\sum_{i \neq \ell_j}
	\int_{0}^\delta (f(A_j(i,\theta)) - f(A_{j-1}(i,\theta)) \right)
	\quad (\mbox{Use Lemma~\ref{lem:charging2}})\\
	&\ge \rho_0 +  \int_0^\delta f(A(\theta)) d\theta -  \delta
	f(\emptyset) + \int_0^1 f(U(\theta)) d\theta - f(\emptyset) \quad
	(\mbox{Use Corollary~\ref{cor:charging1}})\\
	&\ge (k\delta - \delta - 1) f(\emptyset) + \int_0^\delta f(A(\theta))
	d\theta + \int_0^1 f(U(\theta)) d\theta.
  \end{align*}
  We used $\rho_0 = \delta k f(\emptyset)$ in the final inequality.
\end{proofof}

\subsection{Proof of Lemma~\ref{lem:charging1}}
\label{sec:charging1}

Recall that the lemma states that for $\delta \in [1/2,1]$ and for any
$j$,
  $$\sum_{i \neq \ell_j}  \int_{0}^\delta (f(A_j(i,\theta)) -
  f(A_{j-1}(i,\theta))d\theta \ge f(V_j) - f(V_{j-1}) + \alpha_j
  (f(V_{j-1}) - f(V_j)).$$

\begin{proofof}{Lemma~\ref{lem:charging1}}
	Fix $j$ and label $i$. We have $$\int_0^{\delta} (f(A_j(i,
	\theta)) -f(A_{j - 1}(i, \theta))d\theta =
	\int_0^{\min(\delta,x(v_j,i))} (f(A_j(i, \theta))-f(A_{j - 1}(i,
	\theta))d\theta$$ since $A_j(i,\theta) = A_{j-1}(i,\theta)$ when
	$\theta$ is in the interval $(\min(\delta,x(v_j,i)), \delta]$ (or
	the interval is empty).  When $\theta \le x(v_j,i)$ we have
	$A_{j}(i,\theta) = A_{j-1}(i,\theta) + v_j$. Since $f$ is
	submodular and $A_{j - 1}(i, \theta) \subseteq V_{j - 1}$, it
	follows that, for any $\theta \leq x(v_j, i)$, we have
		$$f(A_j(i,\theta)) - f(A_{j-1}(i,\theta)) =
		f(A_{j-1}(i,\theta)+v_j) - f(A_{j-1}(i,\theta)) \ge f(V_{j-1}
		+ v_j) - f(V_{j-1}) = f(V_j) - f(V_{j-1}).$$
	Therefore,
	\begin{align*}
		\int_0^{\delta} (f(A_j(i, \theta)) -f(A_{j - 1}(i,
		\theta))d\theta
		& = \int_0^{\min(\delta,x(v_j,i))} (f(A_j(i, \theta))-f(A_{j -
		1}(i,\theta))d\theta \\
		& \ge \int_0^{\min(\delta,x(v_j,i))} (f(V_j) - f(V_{j-1}))
		d\theta\\
		& = \min(\delta, x(v_j, i)) (f(V_j) - f(V_{j - 1})).
	\end{align*}
	Note that, for any $i \neq \ell_j$, $x(v_j, i) \leq \delta$: if
	$\alpha_j \leq \delta$, the claim follows, since $x(v_j, i) \leq
	\alpha_j$; otherwise, since $\delta \geq 1/2$ and $\sum_i x(v_j,
	i) = 1$, it follows that $x(v_j, i) \leq \delta$ for all $i \neq
	\ell_j$. Therefore, by using the previous bound,
    \begin{align*}
		\sum_{i \neq \ell_j} \int_0^{\delta} (f(A_j(i, \theta) -
		f(A_{j-1}(i, \theta))d\theta
		& \ge \sum_{i \neq \ell_j} \min(\delta, x(v_j, i)) (f(V_j) -
		f(V_{j - 1})) \\
 		& = \sum_{i \neq \ell_j} x(v_j, i) (f(V_j) - f(V_{j - 1})) \\
 		& = (1- x(v_j, \ell_j)) (f(V_j) - f(V_{j - 1})) \\
 		& =  f(V_j) - f(V_{j - 1}) + \alpha_j (f(V_{j - 1}) - f(V_j)).
	\end{align*}
\end{proofof}

\subsection{Proof of Lemma~\ref{lem:charging2}}
\label{sec:charging2}

\noindent
We recall the statement of the lemma.  Let $\delta \in [0,1]$ and let
$h$ be the largest value of $j$ such that $\alpha_j \le \delta$. Then
  $$\sum_{j=1}^n  \left (\int_{0}^\delta (f(A_j(\ell_j,\theta)) -
  f(A_{j-1}(\ell_j,\theta))d\theta\right) \ge \int_0^\delta
  f(A(\theta)) d\theta -  \delta f(\emptyset).$$
Our goal is to obtain a suitable expression that is upper bounded by
the quantity  $\int_{0}^\delta (f(A_j(\ell_j,\theta)) -
f(A_{j-1}(\ell_j,\theta))d\theta$. It turns out that this expression
has several terms and when we sum over all $j$ they telescope to give
us the desired bound.

We begin by simplifying $\int_{0}^\delta (f(A_j(\ell_j,\theta)) -
f(A_{j-1}(\ell_j,\theta))d\theta$ by applying submodularity. The
following proposition follows from the fact that $f$ is submodular and
$A_j(\ell_j, \theta) \subseteq A_j(\theta)$.

\begin{prop} \label{prop:charging2}
	For any $j$ such that $1 \leq j \leq n$,
	$$\int_0^{\delta}\left(f(A_j(\ell_j, \theta)) - f(A_{j -
	1}(\ell_j, \theta))\right)d\theta \geq \int_0^{\min(\alpha_j,
	\delta)}\left(f(A_{j - 1}(\theta) + v_j) - f(A_{j -
	1}(\theta))\right)d\theta.$$
\end{prop}
\begin{proof}
	If $\theta \in [0,\min(\delta,\alpha_j)]$ we have
	$A_j(\ell_j,\theta) = A_{j-1}(\ell_j,\theta) + v_j$. If $\theta
	\in (\min(\delta,\alpha_j), \delta]$ then $A_j(\ell_j,\theta) =
	A_{j-1}(\ell_j,\theta)$.  Therefore
	\begin{align*}
		\int_0^{\delta}\left(f(A_j(\ell_j, \theta)) - f(A_{j -
		1}(\ell_j, \theta))\right)d\theta
		& = \int_0^{\min(\delta, \alpha_j)} \left(f(A_j(\ell_j,
		\theta)) - f(A_{j - 1}(\ell_j, \theta))\right)d\theta\\
		& =  \int_0^{\min(\delta, \alpha_j)} \left(f(A_{j - 1}(\ell_j,
		\theta) + v_j) - f(A_{j - 1}(\ell_j, \theta))\right)d\theta
	\end{align*}
	Since $f$ is submodular and $A_{j-1}(\ell_j,\theta) \subseteq
	A_{j-1}(\theta)$, it follows that, for any $\theta \leq \alpha_j$,
		$$f(A_j(\ell_j, \theta)) - f(A_{j - 1}(\ell_j, \theta)) \ge
		f(A_{j - 1}(\theta) + v_j) - f(A_{j - 1}(\theta))$$
	and the proposition follows.
\end{proof}

\bigskip
\noindent
Let $\Delta_j = \int_0^{\alpha_j} \left(f(A_{j - 1}(\theta) + v_j) -
f(A_{j - 1}(\theta))\right)d\theta$, and $\Lambda_j =
\int_{\delta}^{\alpha_j} \left(f(A_{j - 1}(\theta) + v_j) - f(A_{j -
1}(\theta))\right)d\theta$. Note that the right hand side of the
inequality in Proposition~\ref{prop:charging2} is equal to $\Delta_j$
if $j \leq h$, and it is equal to $\Delta_j - \Lambda_j$ otherwise.
Proposition~\ref{prop:Delta} and Proposition~\ref{prop:Lambda} express
$\Delta_j$ and $\Lambda_j$ in a more convenient form.

Let $V_{j', j} = \{v_{j'}, v_{j' + 1}, \cdots, v_j\}$ for all $j'$ and
$j$ such that $j' \leq j$; let $V_{j', j} = \emptyset$ for all $j'$
and $j$ such that $j' > j$.

\begin{prop} \label{prop:Delta}
	$$\Delta_j = \sum_{j' = 1}^n (\alpha_{j'} - \alpha_{j' - 1})
	(f(V_{j', j}) - f(V_{j', j - 1})).$$
\end{prop}
\begin{proof}
	It follows Proposition~\ref{prop:allocated-helper} that, if
	$\theta \in (\alpha_{j' - 1}, \alpha_{j'}]$, $A(\theta) = V_{j',
	n}$ and $A_j(\theta) = V_{j', j}$. Therefore
	\begin{align*}
		\Delta_j = \int_0^{\alpha_j}
		(f(A_{j-1}(\theta)+v_j)-f(A_{j-1}(\theta))d\theta
		&= \sum_{j' = 1}^j \int_{\alpha_{j' - 1}}^{\alpha_{j'}}
		\left(f\left(A_{j - 1}(\theta) + v_j\right) - f\left( A_{j -
		1}(\theta) \right)\right) d\theta\\
		&= \sum_{j' = 1}^j (\alpha_{j'} - \alpha_{j' - 1}) (f(V_{j',
		j}) - f(V_{j', j - 1}))\\
		&= \sum_{j' = 1}^n (\alpha_{j'} - \alpha_{j' - 1}) (f(V_{j',
		j}) - f(V_{j', j - 1})).
	\end{align*}
	The last line follows from the fact that, if $j' > j$, $V_{j', j}
	= V_{j', j - 1} = \emptyset$.
\end{proof}

\noindent \medskip
The corollary below follows by simple algebraic manipulation and is moved
to Appendix~\ref{app:charging2}.

\begin{corollary} \label{cor:Delta}
	$$\sum_{j = 1}^n \Delta_j = \sum_{j = 1}^n \alpha_j (f(V - V_{j -
	1}) - f(V - V_j)).$$
\end{corollary}

\noindent \medskip
We now consider $\Gamma_j$. 

\begin{prop} \label{prop:Lambda}
	For all $j > h$ where $h$ is the largest index such that $\alpha_h
	\le \delta$,
	$$\Lambda_j = (\alpha_{h+1} - \delta) (f(V_{h+1, j}) - f(V_{h, j -
	1})) +  \sum_{j' = h + 2}^n (\alpha_{j'} - \alpha_{j' -
	1})(f(V_{j', j}) - f(V_{j', j - 1})).$$
\end{prop}
\begin{proof}
  For notational convenience let $\beta_h = \delta$ and $\beta_j =
  \alpha_j$ for all $j > h$. It follows
  Proposition~\ref{prop:allocated-helper} that, if $\theta \in
  (\beta_{j' - 1}, \beta_{j'}]$, $A(\theta) = V_{j', n}$ and
  $A_j(\theta) = V_{j', j}$. Therefore

	\begin{eqnarray*}
		\int_{\delta}^{\alpha_j} (f(A_{j - 1}(\theta) + v_j) - f(A_{j
		- 1}(\theta)))d\theta &=& \sum_{j' = h + 1}^j \int_{\beta_{j'
		- 1}}^{\beta_{j'}}(f(A_{j - 1}(\theta) + v_j) - f(A_{j -
		1}(\theta)))d\theta\\
		&=& \sum_{j' = h + 1}^j (\beta_{j'} - \beta_{j' - 1})(f(V_{j',
		j}) - f(V_{j', j - 1}))\\
		&=& \sum_{j' = h + 1}^n (\beta_{j'} - \beta_{j' - 1})(f(V_{j',
		j}) - f(V_{j', j - 1})).
	\end{eqnarray*}
	The last line follows from the fact that, if $j' > j$, $V_{j', j}
	= V_{j', j - 1} = \emptyset$. The lemma follows by noting that
    \begin{align*}
		\sum_{j' = h + 1}^n (\beta_{j'} - \beta_{j' - 1})(f(V_{j', j})
		- f(V_{j', j - 1}))
		& = (\alpha_{h+1} - \delta) (f(V_{h+1, j}) - f(V_{h, j - 1}))\\
		& +  \sum_{j' = h + 2}^n (\alpha_{j'} - \alpha_{j' -
		1})(f(V_{j', j}) - f(V_{j', j - 1})).
	\end{align*}
\end{proof}

\noindent \medskip
The corollary below follows by simple algebraic manipulation and is moved
to Appendix~\ref{app:charging2}.
\begin{corollary} \label{cor:Lambda}
	$$\sum_{j = h + 1}^n \Lambda_j = \sum_{j = h + 1}^n \alpha_j(f(V -
	V_{j - 1}) - f(V - V_j)) - \delta(f(V - V_h) - f(\emptyset)).$$
\end{corollary}

\noindent \medskip
Now we finish the proof.

\begin{proofof}{Lemma~\ref{lem:charging2}}
We apply Proposition~\ref{prop:charging2} in the first inequality below,
and then Corollary~\ref{cor:Delta} and Corollary~\ref{cor:Lambda} 
to derive the third line from the second.
\begin{eqnarray*}
  \sum_{j = 1}^n \int_0^{\delta}\left( f(A_j(\ell_j, \theta)) -
    f(A_{j - 1}(\ell_j, \theta)) \right)d\theta
  &\geq& \sum_{j = 1}^n \int_0^{\min(\alpha_j,
    \delta)}\left(f(A_{j - 1}(\theta) + v_j) - f(A_{j -
      1}(\theta))\right)d\theta\\
  &=& \sum_{j = 1}^n \Delta_j - \sum_{j = h + 1}^n \Lambda_j\\
  &=& \sum_{j = 1}^h \alpha_j (f(V - V_{j - 1}) - f(V - V_j)) +
  \delta(f(V - V_h) - f(\emptyset))\\
  &=& \int_0^{\delta} f(A(\theta)) d\theta - \delta f(\emptyset).
\end{eqnarray*}
The last equality follows from Proposition~\ref{prop:allocated}.
\end{proofof}

\section{Conclusions and Open Problems}
The main open question is whether the integrality gap of \SubMPRel for
\SymSubMP is stricly smaller than the bound of $1.5-1/k$ we showed in
this paper. Karger \etal \cite{KargerKSTY99} rely extensively on the
geometry of the simplex to obtain a bound of $1.3438$ for \MC via the
relaxation from \cite{CalinescuKR98}. However, we mention that the
rounding algorithms used in \cite{KargerKSTY99} have natural analogues
for rounding \SubMPRel but analyzing them is quite challenging 
for an arbitrary symmetric submodular function.

Zhao, Nagamochi and Ibaraki \cite{ZhaoNI05} considered a common
generalization of \SubMP and \kSubMP where we are given a set $S$ of
terminals with $|S| \ge k$ and the goal is to partition $V$ into $k$
sets $A_1,\ldots, A_k$ such that each $A_i$ contains at least one
terminal and $\sum_{i=1}^k f(A_i)$ is minimized. Note that when $|S| =
k$ we get $\SubMP$ and when $S = V$ we get $\kSubMP$. The advantage
of the greedy splitting algorithms developed in \cite{ZhaoNI05} is
that they extend to these more general problems.  However, unlike the
case of \SubMP, there does not appear to be an easy way to write a
relaxation for this more general problem; in the special case
of graphs such a relaxation has been developed \cite{ChekuriGN06}.
An important open problem here is whether the $k$-way cut problem
in graphs admits an approximation better than $2(1-1/k)$. 

Related to the above questions is the complexity of \kSubMP when $k$
is a fixed constant.  For \SymSubMP a polynomial-time algorithm was
claimed in \cite{Queyranne99} although no formal proof has been
published; this generalizes the polynomial-time algorithm for graph
$k$-cut problem first developed by Goldschmidt and Hochbaum
\cite{GoldschmidtH94}. There has been particular interest in the
special case of \kSubMP, namely, the hypergraph $k$-cut problem; a
polynomial time algorithm for $k=3$ was developed in \cite{Xiao10}
and extended to \SubMP in \cite{OkumotoFN10}. Fukunaga \cite{Fukunaga10}
gave polynomial time algorithms when $k$ and the maximum hyperedge size
are both fixed. The following is an an open problem. Does
the hypergraph $k$-cut problem for $k=4$ have a polynomial time algorithm?


\newpage
\appendix

\section{Omitted proofs from Section~\ref{sec:charging2}}
\label{app:charging2}

\begin{proofof}{Corollary~\ref{cor:Delta}}
	It follows from Proposition~\ref{prop:Delta} that
	\begin{eqnarray*}
		\sum_{j = 1}^n \Delta_j &=& \sum_{j' = 1}^n (\alpha_{j'} -
		\alpha_{j' - 1}) \sum_{j = 1}^n (f(V_{j', j}) - f(V_{j', j -
		1}))\\
		&=& \sum_{j' = 1}^n (\alpha_{j'} - \alpha_{j' - 1}) (f(V_{j',
		n}) - f(V_{j', 0}))\\
		&=& \sum_{j' = 1}^n (\alpha_{j'} - \alpha_{j' - 1}) (f(V_{j',
		n}) - f(\emptyset))\\
		&=& \sum_{j' = 1}^n \alpha_{j'} (f(V_{j', n}) - f(\emptyset))
		- \sum_{j' = 0}^{n - 1} \alpha_{j'} (f(V_{j' + 1, n}) -
		f(\emptyset))\\
		&=& \sum_{j' = 1}^n \alpha_{j'} (f(V_{j', n}) - f(V_{j' + 1,
		n})) + \alpha_n (f(V_{n + 1, n}) - f(\emptyset))\\
		&=& \sum_{j' = 1}^n \alpha_{j'} (f(V_{j', n}) - f(V_{j' + 1,
		n}))\\
		&=& \sum_{j' = 1}^n \alpha_{j'}(f(V - V_{j' - 1}) - f(V -
		V_{j'}))
	\end{eqnarray*}
\end{proofof}

\begin{proofof}{Corollary~\ref{cor:Lambda}}
	For notational convenience, let $\beta_h = \delta$ and $\beta_j =
	\alpha_j$ for all $j > h$.  It follows from
	Proposition~\ref{prop:Lambda} that
	\begin{eqnarray*}
		\sum_{j = h + 1}^n \Lambda_j
		&=& \sum_{j' = h + 1}^n (\beta_{j'} - \beta_{j' - 1}) \sum_{j
		= h + 1}^n (f(V_{j', j}) - f(V_{j', j - 1}))\\
		&=& \sum_{j' = h + 1}^n (\beta_{j'} - \beta_{j' - 1})(f(V_{j',
		n}) - f(V_{j', h}))\\
		&=& \sum_{j' = h + 1}^n (\beta_{j'} - \beta_{j' - 1})(f(V_{j',
		n}) - f(\emptyset))\\
		&=& \sum_{j' = h + 1}^n \beta_{j'} (f(V_{j', n}) - f(\emptyset)) -
		\sum_{j' = h}^{n - 1}\beta_{j'} (f(V_{j' + 1, n}) - f(\emptyset))\\
		&=& \sum_{j' = h + 1}^n \beta_{j'} (f(V_{j', n}) - f(V_{j' +
		1, n})) - \beta_h (f(V_{h + 1}, n) - f(\emptyset)) + \beta_n
		(f(V_{n + 1, n}) - f(\emptyset))\\
		&=& \sum_{j' = h + 1}^n \alpha_{j'} (f(V - V_{j' - 1}) - f(V -
		V_{j'})) - \delta (f(V - V_h) - f(\emptyset))
	\end{eqnarray*}
\end{proofof}


\end{document}